\documentclass[letterpaper, 10 pt, conference]{ieeeconf}  
\usepackage[normalem]{ulem}
\usepackage{subcaption}
\usepackage{graphicx}
\usepackage{algorithmic,algorithm}

\IEEEoverridecommandlockouts                              
\usepackage{cite}
\usepackage{times}
\usepackage{setspace}
\usepackage{url}
\spacing{1}
\usepackage[utf8]{inputenc}
\usepackage[T1]{fontenc}
\let\labelindent\relax

\usepackage{graphicx}		
\usepackage{subcaption}
\usepackage{wrapfig}
\usepackage[format=plain,font=footnotesize,labelfont=bf,labelsep=period]{caption}
\usepackage{sidecap} 
\usepackage[export]{adjustbox}
\usepackage[font=small]{caption}
\usepackage{float}
\usepackage{dblfloatfix}

\usepackage{amsmath} 
\usepackage{amssymb}  
\usepackage{amsthm}
\usepackage{mathtools}
\usepackage[normalem]{ulem}
\usepackage{paralist}	
\usepackage[space]{grffile} 
\usepackage{color}
\usepackage{bbm}
\usepackage{placeins} 
\usepackage{array}
\usepackage{siunitx}
\usepackage{hyperref}
\usepackage{enumitem}                               
\usepackage{cleveref}

\newtheorem{theorem}{Theorem}

\newtheorem{lemma}{Lemma}
\theoremstyle{definition}
\newtheorem{definition}{Definition}
\theoremstyle{definition}
\newtheorem{remark}{Remark}
\theoremstyle{definition}
\newtheorem{assumption}{Assumption}
\crefname{assumption}{Assumption}{Assumptions}
\theoremstyle{definition}

\theoremstyle{definition}
\newtheorem{example}{Example}

\newcommand{\X}{\mathcal{X}}
\newcommand{\U}{\mathcal{U}}

\newcommand{\R}{\mathbb{R}}
\newcommand{\C}{\mathcal{C}}

\newcommand{\mb}[1]{\mathbf{ #1 }}
\newcommand{\bs}[1]{\boldsymbol{ #1 }}

\DeclareMathOperator*{\argmin}{argmin}

\newcommand{\lmat}{\begin{bmatrix}}
\newcommand{\rmat}{\end{bmatrix}}
\title{\LARGE \bf
Safe Online Dynamics Learning with Initially Unknown Models and Infeasible Safety Certificates}

 \author{Alexandre Capone$^1$, Ryan K. Cosner$^2$, Aaron D. Ames$^2$, and Sandra Hirche$^1$
 \thanks{$^{1}$Alexandre Capone and Sandra Hirche are with the Department of Electrical and Computer Engineering, Technical University of Munich, 80333 Munich, Germany. Emails: $\{\texttt{alexandre.capone, hirche}\}$@tum.de).}%
 \thanks{$^{2}$ Ryan K. Cosner and Aaron D. Ames are with the Department of Mechanical and Civil Engineering, California Institute of Technology,
         Pasadena, CA 91125, USA. Emails:
         $\{\texttt{rkcosner, ames}\}$@caltech.edu}}

\begin{document}

\maketitle
\thispagestyle{empty}
\pagestyle{empty}

\begin{abstract}
Control tasks with high levels of uncertainty and safety requirements are increasingly common. Typically, techniques that guarantee safety during learning and control utilize constraint-based safety certificates, which can be leveraged to compute safe control inputs. However, if model uncertainty is very high, the corresponding certificates are potentially invalid, meaning no control input satisfies the constraints imposed by the safety certificate. This paper considers a learning-based setting with a safety certificate based on a control barrier function second-order cone program. If the control barrier function certificate is valid, our approach leverages the control barrier function to guarantee safety. Otherwise, our method explores the system to recover the feasibility of the control barrier function constraint as fast as possible. To this end, we employ a method inspired by well-established tools from Bayesian optimization. We show that if the sampling frequency is high enough, we recover the feasibility of a control barrier function second-order cone program, guaranteeing safety. To the best of our knowledge, this corresponds to the first algorithm that guarantees safety through online learning without requiring a prior model or backup safe non-learning-based controller. 

\vspace{1em}

\end{abstract}


\section{INTRODUCTION}
With the growing proliferation of robotics in safety-critical fields, e.g., autonomous vehicles, medical robotics, and aerospace systems, the need for methods which ensure the safety of systems and their users has become paramount. In control theory, guaranteeing safety has become synonymous with guaranteeing the forward invariance of some user-defined \textit{safe set} \cite{ames_control_2017}. To achieve this, several methods have been developed, including  Model Predictive Control (MPC) with state constraints \cite{wabersich_linear_2018}, Reachability-based methods \cite{bansal2017hamilton}, and Control Barrier Functions (CBFs) \cite{ames_control_2017}, all of which are capable of providing rigorous mathematical guarantees of safety. 

Unfortunately, the safety guarantees of the previously mentioned methods typically rely on an assumed perfect knowledge of the system dynamics, which is not necessarily available in practice. In general, dynamics models have some associated errors that must be accounted for to achieve safety. This error is often modeled using data-driven frameworks, such as neural networks \cite{taylor_learning_nodate,csomay-shanklin_episodic_nodate} or Gaussian processes (GPs) \cite{lederer2022safe,rodriguez2021learning}. In order to still be able to guarantee safety whenever a residual model is involved, accurate model error bounds are generally required and need to be carefully considered in the control design. This has been carried out, e.g., for MPC \cite{hewing2020learning,koller2018learning}, reachability-based techniques \cite{akametalu2015reachability}, and control-barrier functions \cite{castaneda2021pointwise,jagtap2020control}. However, a considerable drawback of the aforementioned techniques is that they all rely on the existence of a (potentially conservative) backup safe controller that they can fall back on in case the system reaches a critical state that the main algorithm cannot address.

In this work, we present an online learning-based algorithm that is able to guarantee feasibility without requiring an a-priori safe controller. Our approach learns the time-derivative of a Control Barrier function online and attempts to solve an associated quadratic program (QP) that takes the model error into account. Whenever the QP becomes infeasible, our approach leverages techniques from Bayesian optimization to ensure that the time-derivative of the Control Barrier Function is learned fast enough in order to recover feasibility of the associated SOCP before the system becomes unsafe. Our method is applicable for a rich class of potential systems and only requires the knowledge of a CBF for the system. To the best of our knowledge, this represents the first approach that allows for safe control without any prior system model or backup safe control law.

The remainder of this paper is structured as follows\footnote{\textbf{Notation:}For scalars $a_1,...a_N$, we use the notation $\text{diag}(a_1,\ldots,a_N)$ to refer to the $N$-dimensional diagonal matrix with $a_1,...a_N$ as diagonal entries. We employ $\vert \cdot \vert$ to denote the determinant operator, and $\lVert\cdot\rVert_2$ to refer to the $2$-norm of both vectors and matrices. For any set $\mathcal{S}$, we employ $\mathcal{P}(\mathcal{S})$ to refer to the power set of $\mathcal{S}$.}. In \Cref{section:problemsetting} we describe the problem setting considered in this paper, together with some background on control barrier functions and Gaussian processes. Our main contribution, which includes the recovery of feasibility of a CBF-SOCP by means of online learning, is presented in \Cref{section:mainresult}. In \Cref{sect:numericalvalidatoin}, we provide a numerical validation of our approach on using a cruise control and a quadcopter model. 
We finalize the paper with some conclusions, in \Cref{section:conclusion}.




\section{Background and Problem Setting}
\label{section:problemsetting}
 
Consider the control affine system
\begin{align}
    \dot{x} = f(x) + g(x) u  \label{eq:ol_dyn}
\end{align}
where $x \in \mathcal{X} \subseteq \R^n $ and $u \in \mathcal{U} \subseteq \R^m$, and $f: \R^n \to \R^n $ and $g: \R^n \to \R^{n \times m } $ are (partially) unknown locally Lipschitz continuous functions that represent the drift dynamics and the input matrix, respectively. We further assume that $\mathcal{X}$ is an open and connected set and that $\mathcal{U}$ is compact. When a locally Lipschitz controller $\pi: \R^n \to \R^m $ is used, we can define the closed-loop system:
\begin{align}
\dot{x} = f(x) + g(x) \pi(x) \label{eq:cl_dyn}
\end{align}
whose solutions we assume are forward complete (i.e., exist for all $t\geq 0 $). We also assume to know an upper bound $M_{\dot{x}}\in \mathbb{R}_+$ for the time-derivative of the dynamics, as specified by the following assumption. i.e., 
\begin{assumption}
\label{assumption:knownboundontimederivative}
    There exists a known positive constant $M_{\dot{x}}$, such that $\Vert \dot{x}\Vert_2\leq M_{\dot{x}}$ holds for all $ x\in \X$ and $u \in \U$.
\end{assumption}

To define safety we consider a \textit{safe set} $\mathcal{C}$, which we aim to render forward invariant, as specified in the following. 

\begin{definition}[Forward Invariance (Safety)]
    A set $\mathcal{C}\subset \mathcal{X}$ is \textit{forward invariant} for system \eqref{eq:cl_dyn}  if $x(0) \in \mathcal{C}$ implies that $x(t) \in \mathcal{C}$ for all $t \geq 0 $. The closed loop system \eqref{eq:cl_dyn} is said to be \textit{safe} with respect to $\mathcal{C}$ if $\mathcal{C}$ is forward invariant. 
\end{definition}

Our goal is to design a control law $\pi$ that renders the closed loop system \eqref{eq:cl_dyn} safe. To this end, we consider the case where the safe set $\mathcal{C}$ corresponds to the superlevel set of some known continuously differentiable function $h: \R^n \to \R$ with 0 a regular value \footnote{A function $h$ has 0 as a \textbf{\textit{regular value}} if $h(x) = 0  \implies \frac{\partial h}{\partial x}(x) \neq 0 $. }: 
\begin{align*}
\label{eq:specificatoinofC}
    \C & \triangleq \{ x \in \R^n ~|~ h(x) \geq 0 \}. 
\end{align*}
This specification of $\mathcal{C}$
allows us to achieve safety through the control barrier function framework, which we discuss in the following section.

\subsection{Safety through Control Barrier Functions}

Before introducing control barrier functions, we briefly introduce the concepts of class $\mathcal{K}_\infty$ ($\mathcal{K}_{\infty}^e)$ function and barrier functions. 

We call a continuous function $\alpha : \R\to \R$ an extended class $\mathcal{K}_\infty$ ($\mathcal{K}_{\infty}^e)$ if it is strictly monotonically increasing and satisfies $\alpha(0) = 0$ and is radially unbounded. Barrier Functions (BFs) can be used to synthesize controllers ensuring the safety of the closed-loop system \eqref{eq:cl_dyn} with respect to a given set $\mathcal{C}$. 

\begin{definition}[Barrier Function (BF) \cite{ames_control_2017}]
Let $\mathcal{C}\subset\mathcal{X}$ be the 0-superlevel set of a continuously differentiable function $h: \R^n \to \R$ with zero a regular value. The function $h$ is a barrier function (BF) for \eqref{eq:cl_dyn} if there exists some $ \alpha \in \mathcal{K}_{\infty}^{e}$ such that
\begin{align}
    \frac{dh}{dt}(x) = \frac{\partial h}{\partial x} (x) ( f(x) + g(x) \pi(x)) \geq -\alpha(h(x)), \quad \forall x \in \mathcal{X}
\end{align}
\end{definition}

Barrier functions have the following associated safety guarantee: 
\begin{lemma}[{[\citenum{ames_control_2017}]}]
    Let $\mathcal{C}\subset \R^n $ be the 0-superlevel set of a continuously differentiable function $h:\R^n \to \R$ with $\frac{\partial h}{\partial x}(x) \neq 0 $ when $h(x) =0$. If $h$ is a BF for \eqref{eq:cl_dyn} on $\mathcal{C}$, then the system \eqref{eq:cl_dyn} is safe with respect to $\mathcal{C}$. 
\end{lemma}

Furthermore, control barrier functions (CBFs) serve a similar role for system \eqref{eq:ol_dyn} without a pre-defined controller and can be used to generate safe control inputs. 

\begin{definition}[Control Barrier Function (CBF) \cite{ames_control_2017}]
    Let $\mathcal{C} \subset\mathcal{X} $ be the 0-superlevel set of a continuously differentiable function $h: \R^n \to \R$  with 0 a regular value. The function $h$ is a control barrier function (CBF) for \eqref{eq:ol_dyn} on $\mathcal{C}$ if there exists an $\alpha \in \mathcal{K}_\infty^e$ such that for all $x \in \mathcal{X}$: 
    \begin{align}
    \label{eq:cbf_constraint}
    \begin{split}
        \sup_{u \in \mathcal{U}} \dot{h}(x, u)  \triangleq  \sup_{u \in \mathcal{U}} \frac{\partial h }{\partial x }(x)\left( f(x) + g(x)u \right)> -\alpha(h(x)) .
        \end{split}
    \end{align}
\end{definition}

In this paper, we consider Lipschitz continuous functions $h$ which yield compact safe sets, as stated in the following assumption.

\begin{assumption}
\label{assumption:Lipschitzh}
    The safe set $\mathcal{C}$ is compact and the function $h$ admits a Lipschitz constant $L_h$, i.e., $\vert h(x) - h(x') \vert \leq L_h \Vert x-x'\Vert_2$ holds for all $x, x' \in \mathcal{X}$. 
\end{assumption}

Given a CBF $h$ and corresponding $\alpha \in \mathcal{K}_\infty^e$, there exist a corresponding non-empty set of safe inputs for each $x \in \mathcal{X} $ given as $\Pi_{\textrm{CBF}}: \mathcal{X} \to \mathcal{P}(\mathcal{U})$: 
\begin{align}
    \Pi_{\textrm{CBF}}(x) \triangleq \left \{ u \in \mathcal{U} ~\Big | ~ \dot{h}(x, u) \geq - \alpha (h(x)) \right\}
\end{align}
\noindent which can be used to synthesize controllers which guarantee safety. 
\begin{lemma}
[\cite{ames_control_2017}]
\label{thm:standard_cbf}
    Let $\mathcal{C}\subset \mathcal{X} $ be a 0-superlevel set of a continuously differentiable function $h: \R^n \to \R$ with 0 a regular value. If $h$ is a CBF for \eqref{eq:ol_dyn}, then the set $\Pi_{\textrm{CBF}}(x)$ is non-empty for all $x \in \mathcal{X}$ and additionally if $\pi:\mathcal{X} \to \mathcal{U} $ is a locally Lipschitz controller with $\pi(x) \in \Pi_\textrm{CBF}(x)$ for all $x \in \mathcal{X} $, the closed loop system \eqref{eq:cl_dyn} is safe with respect to $\mathcal{C}$. 
\end{lemma}

One common method for synthesizing controllers which satisfy $\pi(x) \in \Pi_\textrm{CBF}(x) $ for all $x \in \mathcal{X} $ is by employing the CBF condition \eqref{eq:cbf_constraint} as a constraint in a safety-filter which enforces safety while achieving minimal deviation from a nominal controller $\pi_\textrm{nom}: \mathcal{X} \to \mathcal{U}$: 
\begin{align}
    \pi_{\textrm{CBF-QP}}(x) = \argmin_{u \in \mathcal{U}}  & \Vert u - \pi_\textrm{nom}(x) \tag{CBF-QP} \label{eq:cbf-qp}\Vert_2 \\
    \textrm{s.t. } & \frac{\partial h}{\partial x}(x)\left( f(x) + g(x) u \right) \geq - \alpha(h(x)). \nonumber
\end{align}
\noindent Given the affine form of this constraint, this safety filter is a quadratic program (QP) which has a closed form solution \cite{ames_control_2017} and can be solved rapidly enough for online safe controller synthesis \cite{gurriet2018towards} for robotic systems. 


Since we seek to consider systems with unknown models, we assume an additional degree of robustness to the CBF constraint which will allow us to find ``safer'' control actions to compensate for this uncertainty:
\begin{assumption}[Robust CBF Feasibility] \label{assp:robust_cbf_feasibility}
There exists some $ \alpha \in \mathcal{K}_{\infty}^{e}$ and a positive scalar $\epsilon >0$, such that
    \begin{align}
        \sup_{u \in \mathcal{U}} \frac{\partial h}{\partial x}(x)\left( f(x) + g(x) u \right)\geq - \alpha (h(x)) + \epsilon \label{eq:robust_cbf}
    \end{align}
    holds for all $x \in \mathcal{X}$.
\end{assumption}


We note that given the existence of a CBF, this assumption is not highly conservative. This can be be seen through the lens of Input-to-State Safety, as shown in the following example.

\begin{example}
    If the CBF condition \eqref{eq:cbf_constraint} holds on $\mathcal{X}$ with $\alpha(r) = ar$ for some $a > 0$, 
    then the Assumption \ref{assp:robust_cbf_feasibility} holds for any function $h'(x) \triangleq h(x) + \frac{\epsilon}{a} $ with $\epsilon>0$ such that $\mathcal{C}_\epsilon = \{ x \in \R^n ~|~ h'(x) \geq 0 \} \subset \mathcal{X}$. This can be seen by taking
    \begin{align*}
        \sup_{u \in \mathcal{U}} \dot{ h}'(x, u) &=  \sup_{u \in \mathcal{U}}\dot{ h } (x, u) \\
        & = \sup_{u \in \mathcal{U}} \frac{\partial h }{\partial x } (x)\left( f(x) + g(x) u \right) \\
        & \geq  - a h(x) +  \epsilon - \epsilon \\
        & =  - a h'(x)   + \epsilon
    \end{align*}
    Thus we recover \eqref{eq:robust_cbf} for $h'$. We note that $ \mathcal{C} \subset \mathcal{C}_\epsilon \subset \mathcal{X}$, so to attain this robust CBF feasibility, the set $\mathcal{C}$ should be chosen with some allowable margin such that, $h(x) \in [-\frac{\epsilon }{a }, 0] $ does not imply catastrophic failure. 
    \end{example}

\subsection{Gaussian Processes and Reproducing Kernel Hilbert Spaces}

Consider a prior model of the system $\hat{f}$ and $\hat{g}$, which can be set to zero whenever no prior knowledge is available. We then model the residual $f-\hat{f}$ and $g-\hat{g}$ of the unknown functions $f$ and $g$ jointly by using a Gaussian process (GP) model. In order to then obtain theoretical guarantees on the growth of the model error as new data is added, we assume that the residual model belongs to a reproducing kernel Hilbert space (RKHS). This will be codified later in this section. In the following, we review both GPs and RKHSs, and provide some preliminary theoretical results.

\begin{definition}[Gaussian Process (GP) \cite{Rasmussen2006}]
    A Gaussian Process (GP) is a collection of random variables, any finite number of which have (consistent) joint Gaussian distributions. 
\end{definition}
\noindent GPs are fully specified by a prior mean, which we set to zero without loss of generality, and kernel $k: \mathcal{X} \times \mathcal{X} \rightarrow \mathbb{R}$. Though GPs are typically defined for scalar functions, a GP model for the vector-valued functions $f$ and $g$ can easily be obtained by defining a separate GP prior for each dimension. We then employ composite kernels
\begin{align}
\label{eq:compositekernel}
    \begin{split}
        k_i(\mb{z},\mb{z}') \triangleq k_{f_i}(x,x') + \sum_{j=1}^m u_j k_{g_{(i,j)}}(x,x')u_j',
    \end{split}
\end{align}
where $\mb{z} \triangleq (x^\top,u^\top)^\top $, to model the entries of $f(x)+g(x)u$ jointly. Here $k_{f_i}$ and $k_{g_{(i,j)}}$ are kernels that capture the behavior of the individual entries of $f$ and $g$, respectively. In practice, if little is known about the system, then so-called ``universal kernels''  (e.g. Matérn or squared-exponential kernel) are employed, which can approximate continuous functions arbitrarily accurately \cite{micchelli2006universal}. In this paper, we consider the practically relevant squared-exponential kernel
\begin{align}
    k_{f_i}(\mb{z},\mb{z}')&= \sigma_{f_i}^2 \exp\left(-\frac{\Vert \mb{z}-\mb{z}'\Vert_2^2}{2 l_{f_i}^2} \right), \\
    k_{g_{(i,j)}}(\mb{z},\mb{z}')&= \sigma_{g_{(i,j)}}^2 \exp\left(-\frac{\Vert \mb{z}-\mb{z}'\Vert_2^2}{2 l_{g_{(i,j)}}^2} \right),
\end{align}
although our approach also is applicable with other kernels. The reason why we employ a composite kernel \eqref{eq:compositekernel} to model the entries of $f$ and $g$ jointly is because this allows us to leverage measurements of $\dot{x}$ to improve our learned model. 

Consider $N$ noisy measurements $\mathcal{D}_{i,N} \triangleq \{ \mb{z}^{(q)}, y_i^{(q)}\}_{q=1,...,N}$ of the i-th entry of the time-derivative of the state, where \begin{align}
\begin{split}
y_i^{(q)} = &f_i(x^{(q)}) - \hat{f}_i(x^{(q)}) \\
&+ \sum_{j=1}^m \left(g_{i,j}(x^{(q)})-\hat{g}_{i,j}(x^{(q)})\right)u_j^{(q)} + \xi^{(q)}_i
\end{split}
\end{align} and the measurement noise $\xi^{(q)}$ satisfies the following assumption. 

\begin{assumption}
    \label{assumption:noise} 
    The measurement noise $\xi^{(q)}$ is iid zero-mean Gaussian noise with covariance $\sigma_{i,\text{ns}}^2$ for every $i\in\{1,...,n\}$.
\end{assumption}

The posterior of the model at an arbitrary state $\mb{z}^*$ is then given by
    \begin{align*}
    \begin{split}
        & \mu_{i,N}(\mb{z}^*) \triangleq \\
        & \qquad \hat{f}_i(x^*)  + \sum_{j=1}^m \hat{g}_{i,j}(x^*)u_j^*  +\mb{k}_{i,*}^\top\left(\mb{K}_{i,N} + \sigma_{i,\text{ns}}^2 \mb{I}_N\right)^{-1}\mb{y}_i ,\\ 
        & \sigma^2_{i,N}(\mb{z}^*) \triangleq 
        k_i(\mb{z}^*,\mb{z}^*) - \mb{k}_{i,*}^\top\left(\mb{K}_{i,N}  + \sigma_{i,\text{ns}}^2 \mb{I}_N\right)^{-1}\mb{k}_{i,*},
        \end{split}
    \end{align*}
where $\mb{k}_{i,*} \triangleq (k_i(\mb{z}^*,\mb{z}^{(1)}),...,k_i(\mb{z}^*,\mb{z}^{(N)}))^\top$, and the entries of the covariance matrix are given by $[\mb{K}_{i,N}]_{pq} = k_i(\mb{z}^{(p)},\mb{z}^{(q)})$. In our approach, we employ $\mu_{i,N}$ as a model of our system, and $\sigma_{i,N}$ as a measure of uncertainty, which we employ to inform the data-collection trigger.


In the following, we refer to the mean and covariance matrix of the full multivariate Gaussian process model as
\begin{align}
\label{eq:multivariategp}
    \begin{split}
        \bs{\mu}_N(\cdot) &\triangleq (\mu_{{1,N}}(\cdot), \ldots, \mu_{{n,N}}(\cdot))^\top, \\ \bs{\Sigma}_N^2(\cdot)&\triangleq \text{diag}(\sigma^2_{1,N}(\cdot),\ldots,\sigma^2_{n,N}(\cdot)).
    \end{split}
\end{align}
Furthermore, we employ $\mb{z}=(x^\top,u^\top)^\top$ and $\mathcal{D}_N =\{ \mb{z}^{(q)}, \mb{y}^{(q)}\}_{q=1,...,N}$, where $\mb{y}^{(q)}\triangleq (y_1^{(q)},...,y_n^{(q)})^\top$, to refer to concatenated state and input, and the full data set corresponding to $N$ system measurements, respectively.

In the following, we assume that the entries $f_i(x) - \hat{f}_i(x) + \sum_{j=1}^m (g_{i,j}(x) - \hat{g}_{i,j}(x))u_j$ of the residual system dynamics belong to the reproducing kernel Hilbert space generated by the composite kernel \eqref{eq:compositekernel}.
\begin{assumption}
\label{assumption:finiteRKHSnorm}
    For every $i=1,..,n$, the function $f_i(x) - \hat{f}_i(x) + \sum_{j=1}^m (g_{i,j}(x) - \hat{g}_{i,j}(x))u_j$ belongs to the RKHS with reproducing kernel $k_i$, and the corresponding RKHS norm is bounded by a known positive scalar $B_{i} \in (0,  \infty)$.
\end{assumption}

Assumption \ref{assumption:finiteRKHSnorm} implies that the entries of the system dynamics can be expressed as a weighted sum of evaluations of the composite kernel \eqref{eq:compositekernel}. This is not very restrictive, as the corresponding function space is considerably rich.

We can then leverage Assumption \ref{assumption:finiteRKHSnorm} to obtain a bound on the model error.
\begin{lemma}[{[\citenum{Capone2019BacksteppingFP}, Lemma 1]}]
\label{lemma:chowdhury}
Let Assumptions \ref{assumption:noise} and \ref{assumption:finiteRKHSnorm} hold. Then, for any $i=1,...,n$, with probability at least $1-\delta$,
\begin{align}
&\bigg \vert f_i(x^*) + \sum_{j=1}^m g(x^*)u_j^* - \mu_{i,N}(\mb{z}^*) \bigg\vert \leq \beta_{i,N}
\sigma_{i,N}(\mb{z}^*) \label{eq:gp_err_bound}
\end{align}
holds for all $\mb{z}\in \X\times\U$ and all $N\in\mathbb{N}$, where
\begin{align}
    \beta_{i,N} \triangleq B_i + \sigma_{i,\text{ns}}\sqrt{2\left(\gamma_{i,N} + 1 + \ln{(n\delta^{-1})}\right) }
\end{align}
and
    \begin{align}
        \gamma_{i, N}\triangleq \max_{\mb{z}^{(1)},...,\mb{z}^{(N)}} \frac{1}{2}\left\vert \mb{I}_N + \sigma_{i,\text{ns}}^{-1}\mb{K}_{i,N} \right\vert
    \end{align}
    corresponds to the maximal information gain after $N$ rounds of data collection.
\end{lemma}

While Lemma \ref{lemma:chowdhury} allows us to bound the GP model error at any given point $\mb{z}^*$ given the data with high probability, to recover feasibility we need also to understand how $\gamma_{i,N}$ and $
\sigma^2_{i,N}$ change as we add more data. To achieve this, we will also employ the following results.

\begin{lemma}[{[\citenum{Srinivas2012}, Lemma 5.4]}]
\label{lemma:srinivasregretbound}
    The posterior variance satisfies
    \begin{align}
        \sum_{q=1}^N
\sigma^2_{i,q}\left(\mb{z}^{(q-1)}\right) \leq \frac{2}{\ln{\left(1 + \sigma_{i,\text{ns}}^{-2}\right)}} \gamma_{i,N} .
    \end{align}
\end{lemma}

Lemma \ref{lemma:srinivasregretbound} allows us to bound the mean of the error as new data is added. 






In order to put everything together, we now only need to determine how quickly $\gamma_{i,N}$ grows. To this end, we employ the two following results, which imply that the growth of the error bound is logarithmic at worst.

\begin{lemma}[{[\citenum{Srinivas2012}, Theorem 5]}]
\label{lemma:infogainsqkernel}
        Let $k_{f_i}$, $ k_{g_{i,j}}$, be squared-exponential kernels. Then, there exist finite positive constants $C_{f_i}, C_{g_i} \in (0, \infty)$, such that, for $N\in \mathbb{N}$ large enough, 
    \begin{align}
        \max_{\mb{z}^{(1)},...\mb{z}^{(N)}} \frac{1}{2}\left\vert \mb{I}_N + \sigma_{i,\text{ns}}^{-1}\mb{K}_{f_i,N} \right\vert &\leq C_{f_i} \left(\log\left(N\right)\right)^{n+m+1} \\
        \max_{\mb{z}^{(1)},...\mb{z}^{(N)}} \frac{1}{2}\left\vert \mb{I}_N + \sigma_{i,\text{ns}}^{-1}\mb{K}_{g_{i,j},N} \right\vert &\leq C_{g_i} \left(\log\left(N\right)\right)^{n+m+1} ,
    \end{align}
    where the entries of  $\mb{K}_{f_i,N}$ and $\mb{K}_{g_{i,j},N}$ are given by \[\left[\mb{K}_{f_i,N}\right]_{q,l} = k_{f_i}(\mb{z}_q, \mb{z}_l), \qquad \left[\mb{K}_{g_{i,j},N}\right]_{q,l} = k_{g_{i,j}}(\mb{z}_q, \mb{z}_l).\]
\end{lemma}

\begin{lemma}
\label{lemma:growthofgammai}
Let $k_i$ be given as in \eqref{eq:compositekernel}. Then there exists a scalar $C_i$, such that
\begin{align}
    \gamma_{i,N} \leq C_i \left(\log\left(N\right)\right)^{n+m+1}.
\end{align}
\label{lemma:compositekernel}
\end{lemma}
\begin{proof}
    The proof follows directly from \Cref{lemma:infogainsqkernel}, together with Theorem 2 and Theorem 3 in \cite{krause2011contextual}.
\end{proof}

\begin{lemma}
\label{lemma:growth_of_betaN}
    Let $\beta_{N} \triangleq \max_{i}\beta_{i,N}$. Then there exists a scalar $C_{\beta}\in(0, \infty)$, such that, for all $N>1$,
    \begin{align*}
        \beta_N \leq  C_\beta \left(\log\left(N\right)\right)^{\frac{n+m+1}{2}}.
    \end{align*}
\end{lemma}
\begin{proof}
    The proof follows directly from the definition of $\beta_{i,N}$ and \Cref{lemma:growthofgammai}.
\end{proof}

\Cref{lemma:chowdhury,lemma:srinivasregretbound,lemma:infogainsqkernel,lemma:compositekernel,lemma:growth_of_betaN} imply that the GP error bound grows logarithmically, i.e., sublinearly. We will leverage this to show that, if we learn our system fast enough, then the model error becomes small enough to recover feasibility before the boundary of the safe set is reached. We will need the following result to determine the required learning rate, which provides an upper bound on the point when linear growth overtakes logarithmic growth.

\begin{lemma}
\label{lemma:exp_overtakes_linear}
    For $a,b>0$, and any $\lambda$ with 
    \begin{align*}
       \lambda\geq \frac{2}{a}\left( \log(b)-\log\left(\frac{a}{2}\right)\right)
    \end{align*}
    it holds that 
    \begin{align*}
        \exp(a \lambda) \geq b(1+\lambda).
    \end{align*}
\end{lemma}
\begin{proof}
    For any $c>0$, it holds that
    \begin{align*}
       &b\exp(\lambda\exp(-c))   = b\sum_{j=0}^\infty \exp(-cj) \frac{\lambda^j}{j!} \\=& b\sum_{j=0}^\infty \frac{(\lambda\exp(-c) )^j}{j!} 
       \geq b\exp(-c) \sum_{j=0}^\infty \frac{\lambda^j}{j!} \\ \geq & b\exp(-c)  (1+\lambda).
    \end{align*}
    Hence, for $c=-\log(a/2)$, we have
    \begin{align*}
        b(1+\lambda) \leq & \exp\left(\log(b)+\lambda\frac{a}{2}-\log\left(\frac{a}{2}\right)\right) \\
        \leq & \exp\left(\lambda\frac{a}{2}+\lambda\frac{a}{2}\right) =  \exp(a\lambda).
    \end{align*}
    
\end{proof}

\section{Recovering Feasibility with GPs}
\label{section:mainresult}

In this section, we present an algorithm geared towards guaranteeing safety by exploring the state and input space efficiently, such that the feasibility of a CBF-based criterion is guaranteed before the boundary of the safe set is reached, which in turn implies safety. We start by showing that safety is guaranteed whenever a CBF-based criterion is satisfied, then present our algorithm, which ensures safety.

For all $N\in\mathbb{N}$, let $\Pi_N: \mathcal{X} \to \mathcal{P}(\R^m) $ be defined as
\begin{align*}
    \Pi_N(x)  \triangleq\left\{ 
    \begin{array}{l|c}
         &\frac{\partial h }{\partial x } (x)\bs{\mu}_N(x, u) \\
         u \in \mathcal{U} & \geq\\
         &
     - \alpha(h(x)) +  L_h  \beta_N\sqrt{\text{tr}\left(\mb{\Sigma}_N^2(x,u)\right)} \end{array} \right\}
\end{align*}
where $L_h >  0 $ is the global Lipschitz constant of $h$ on $\mathcal{X}$ specified in Assumption \ref{assumption:Lipschitzh}.

Note that, due to the nature of the composite kernels \eqref{eq:compositekernel}, $\bs{\mu}_N(x, u)$ is a linear function of $u$ and $\mb{\Sigma}_N^2(x,u)$ is a diagonal matrix whose entries are positive definite quadratic functions of $u$. Hence, $\Pi_N$ is defined using a second-order cone (SOC) constraint.


The specification of $\Pi_N$ is motivated by the fact that if there exists a locally Lipschitz $\pi: \mathcal{X} \to \mathcal{U}$ such that $\pi(x) \in \Pi_N(x)$ for all $x \in \mathcal{X}$, 
then $\mathcal{C}$ can be rendered forward invariant, i.e., safety is robustly guaranteed for the uncertain system. We show this formally in the following.

\begin{lemma}
\label{lemma:safetyifsetisnonempty}
Let \Cref{assumption:knownboundontimederivative,assumption:Lipschitzh,assp:robust_cbf_feasibility,assumption:noise,assumption:finiteRKHSnorm} hold, let $\Pi_N(x) $ be non-empty for all $x \in \mathcal{X}$ and an arbitrary $N\in\mathbb{N}$, and let $\pi: \mathcal{X} \to \mathcal{U}$ be a locally Lipschitz controller such that $\pi(x)\in \Pi_N(x) $ for all $x \in \mathcal{X}$. Then the closed loop system is safe with respect to $\mathcal{C}$ with probability at least $1-\delta$. 
\end{lemma}
\begin{proof}
By \Cref{lemma:chowdhury}, with probability at least $1-\delta$,
    \begin{align*}
        & \frac{d h}{dt}(x) = \frac{\partial h}{\partial x}(x)\left(f(x) + g(x)\pi(x)\right) \\
         = &\frac{\partial h}{\partial x}(x)\left(\bs{\mu}_N(x,\pi(x))\right) \\
        & + \frac{\partial h}{\partial x}(x)\left(f(x) + g(x)\pi(x) - \bs{\mu}_N(x,\pi(x))\right)\\
        \geq &  \frac{\partial h}{\partial x}(x)\left(\bs{\mu}_N(x,\pi(x))\right)  \\
        &- \left\Vert\frac{\partial h}{\partial x}(x)\right \Vert_2\left\Vert\left(f(x) + g(x)\pi(x) - \bs{\mu}_N(x,\pi(x))\right)\right\Vert_2 \\
       \geq &  \frac{\partial h}{\partial x}(x)\left(\bs{\mu}_N(x,\pi(x))\right) - \left\Vert \frac{\partial h }{\partial x}(x) \right\Vert_2 \beta_N \sqrt{\text{tr}\left(\mb{\Sigma}_N^2(x,u)\right)} \\
        \geq &  -\alpha(h(x)) ,
    \end{align*}
    holds for all $x\in\mathcal{X}$,
    where the last inequality comes from the fact that $\pi(x) \in \Pi_N(x) $ for all $x$. Since $\pi$ satisfies the the properties of Lemma \ref{thm:standard_cbf} with probability $1 - \delta $, the set $\mathcal{C}$ is safe with probability $1 - \delta$. 
\end{proof}

\begin{figure*}[th]
\centering
\begin{subfigure}[b]{0.99\textwidth}
\includegraphics[trim={0pt 479pt 0pt 110pt},clip,width = 0.99\columnwidth]{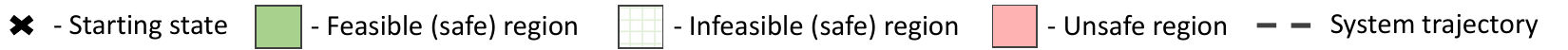}
\end{subfigure}
\begin{subfigure}[b]{0.31\textwidth}
\includegraphics[trim={120pt 90pt 250pt 170pt},clip,width = 0.99\columnwidth]{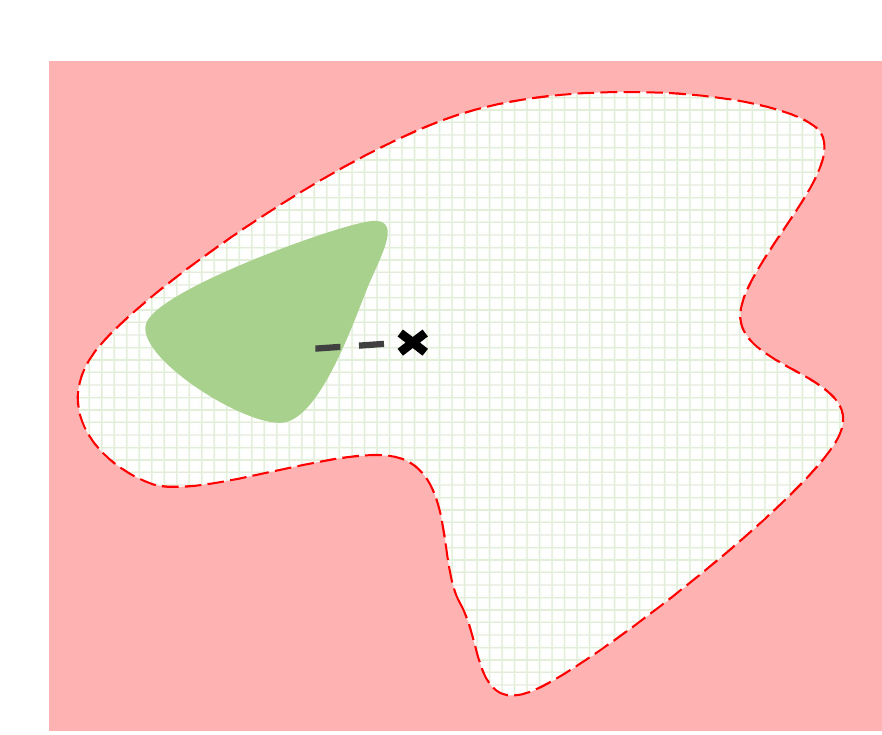}
\caption{System state exits the region where (GP-CBF-SOCP) is infeasible. \\\hspace{\textwidth} \\\hspace{\textwidth}} 
\label{fig:constr_satisf_uniform}
\end{subfigure}
\hfill
\begin{subfigure}[b]{0.31\textwidth}
\includegraphics[trim={120pt 90pt 250pt 170pt},clip,width = 0.99\columnwidth]{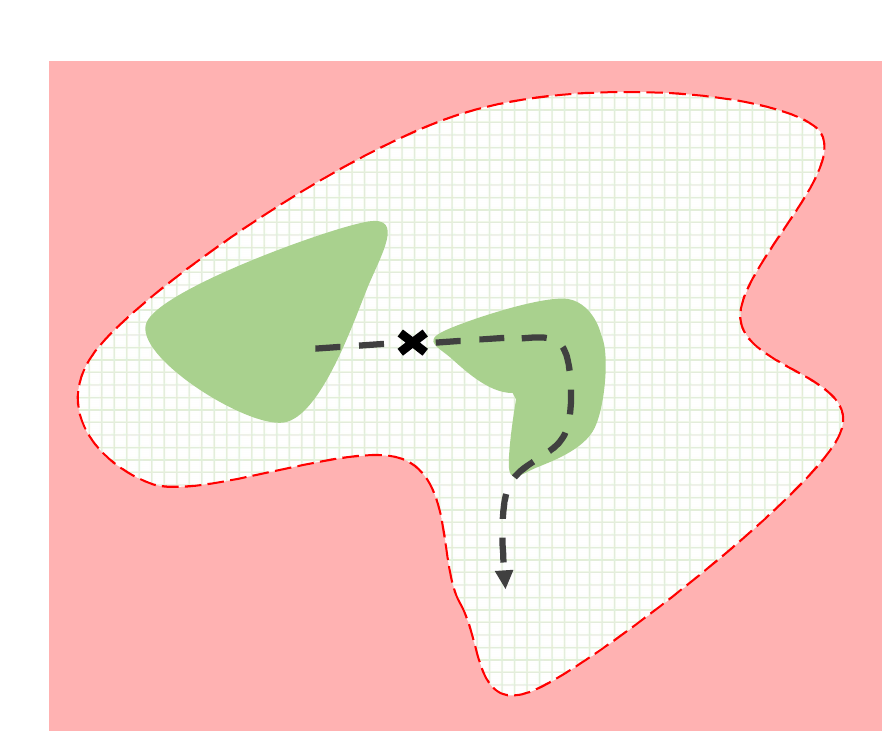}
\caption{Algorithm starts navigating and collecting data by applying \eqref{eq:gp_ucb}. \\\hspace{\textwidth} \\\hspace{\textwidth}} 
\label{fig:cost_uniform}
\end{subfigure}
\hfill
\begin{subfigure}[b]{0.31\textwidth}
\includegraphics[trim={120pt 90pt 250pt 170pt},clip,width = 0.99\columnwidth]{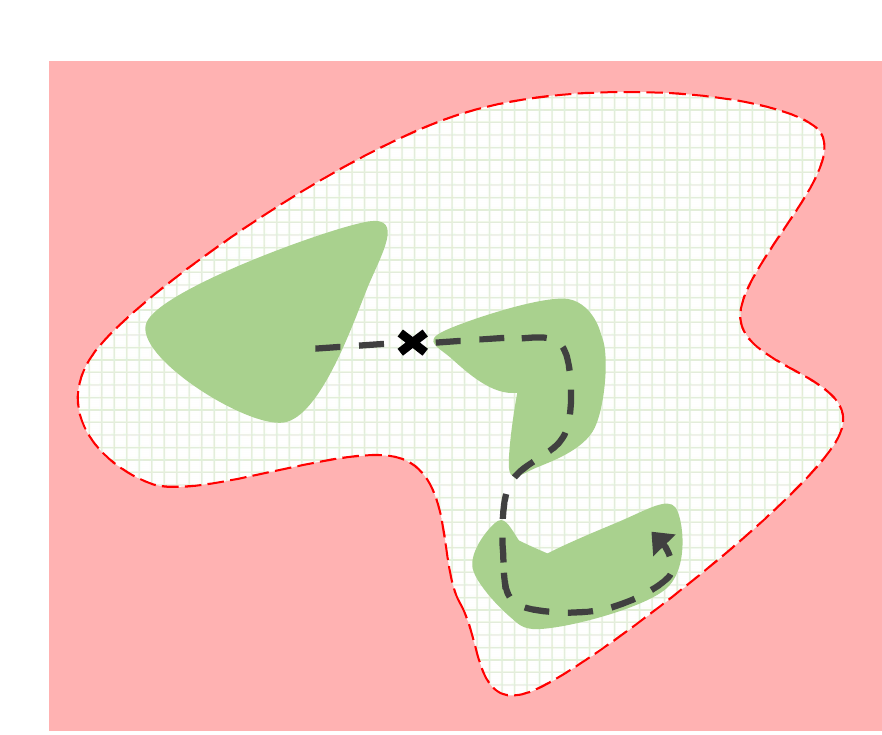}
\caption{\Cref{thm:mainresult} guarantees that the feasibility of (GP-CBF-SOCP) is recovered before the boundary of the safe set is reached, which enables us to compute safe inputs.} 
\label{fig:cost_uniform}
\end{subfigure}
\caption{Illustration of \Cref{alg:feasibility_recovering_algorithm_temporal} and \Cref{thm:mainresult}.}
\label{fig:illustration_of_algorithm}
\end{figure*}



Given a nominal controller $\pi_\textrm{nom}: \mathcal{X} \to \mathcal{U}$, 
we now formulate the robust CBF-SOCP that takes the GP model error into account:
\begin{align}
\begin{matrix}\pi_{\textrm{safe}}(x) = \argmin_{u \in \mathcal{U}}  \Vert u - \pi_\textrm{nom}(x) \Vert_2 \\
   \qquad \hspace{20pt} \textrm{s.t. } u \in \Pi_N(x). \end{matrix} \tag{GP-CBF-SOCP} \label{eq:cbf-qp}\nonumber
\end{align}



We now provide data sampling schemes that recover feasibility of the GP-CBF-SOCP before reaching the boundary of the safe set $\mathcal{C}$. The concept and theoretical guarantees of our approach borrow from the well-known GP upper confidence bound algorithm \cite{Srinivas2012}, which aims to find the maximum of an unknown function by balancing exploration and exploitation. Whenever the SOC constraint is infeasible (i.e., $\Pi_N(x)$ is empty) given the present data set $\mathcal{D}_N$, we propose choosing the input $u$ by solving the optimization problem
\begin{align}
\label{eq:gp_ucb}
\begin{split}
   u^{(N+1)} = & \arg\sup_{u\in\mathcal{U}} \Bigg[ \frac{\partial h }{\partial x } \left(x^{(N+1)}\right)\mu_{N}\left(x^{(N+1)},u\right) \\&+ L_h\beta_{N} \sqrt{\text{tr}\left(\mb{\Sigma}_{N}^2\left(x^{(N+1)},u\right)\right)} \Bigg].
   \end{split}
\end{align}

The idea behind solving \eqref{eq:gp_ucb} is to try and find the control input $u$ that maximizes the unknown time derivative $\frac{dh}{dt}(x)$ of $h(x)$.
\begin{remark}
Note that we are adding the model uncertainty to the model in 
\eqref{eq:gp_ucb}, as opposed to subtracting it as we do in the formulation of the feasible set of the SOCP $\Pi_N(x)$. This is because we want to emphasize exploration over safety whenever the (GP-CBF-SOCP) is infeasible.
\end{remark}
It can be shown that, under the assumption that the state and input space can be sampled arbitrarily, i.e., without constraints imposed by the dynamical system, solving \eqref{eq:gp_ucb} recursively eventually yields a result that is arbitrarily close to the true maximum $ \sup_{u \in \mathcal{U}} \frac{\partial h}{ \partial x} (x)(f(x) + g(x) u)$, eventually recovering feasibility of the (GP-CBF-SOCP) \cite{Srinivas2012}. However, in our setting, we need to recover feasibility before the system reaches the boundary of the safe set $\mathcal{C}$, i.e., we need to sample sufficiently frequently so as to recover feasibility before the boundary of the safe set is reached. To this end, we present an algorithm that samples new data points at pre-specified time intervals in a way that guarantees that $\Delta N$ data points are sampled before the boundary of the safe set $\mathcal{C}$ is reached. This is achieved by leveraging the upper bound on the time-derivative of the system to compute the maximal sampling frequency required to collect at least $\Delta N$ data points before reaching the boundary. This is presented in \Cref{alg:feasibility_recovering_algorithm_temporal}. Note that we also require the control policy $\bar{\pi}(x^{(q)})$ that is applied to the system directly after sampling to be Lipschitz continuous and to correspond to $u^{(q)}$ at $x^{(q)}$. This corresponds to 
Line \ref{codeline:admissibleinputtemporal} in \Cref{alg:feasibility_recovering_algorithm_temporal}.

We can then show that, if we choose the number of data points $\Delta N$ to be collected before reaching the boundary $\partial \mathcal{C}$ high enough, then, with high probability, 
\Cref{alg:feasibility_recovering_algorithm_temporal} always recovers feasibility before reaching the boundary $\partial \mathcal{C}$, guaranteeing safety at all times, and stop collecting data after at most $\Delta N$ points have been collected, meaning that the \Cref{alg:feasibility_recovering_algorithm_temporal} only requires a finite amount of memory.


\begin{algorithm}[b]
\caption{CBF-Control and Sampling Scheme with Guaranteed Feasibility Recovery}
\label{alg:feasibility_recovering_algorithm_temporal}
\begin{algorithmic}[1]
 \renewcommand{\algorithmicrequire}{\textbf{Input:}}
 \renewcommand{\algorithmicensure}{\textbf{Output:}}
 \REQUIRE Size of data set to be collected $\Delta N$
   \IF{$\Pi_N(x)$ is not empty}
  \STATE Solve (GP-CBF-SOCP) and apply $\pi_{\text{safe}}(x)$
  \ELSE
  \STATE Set $\Delta t =  {h(x)} \left({M_{\Vert\frac{\partial h }{ \partial x } \Vert} M_{\dot{x}} \Delta N }\right)^{-1} $ 
  \STATE Set $q = 0$, $t_{\text{infeasible}} = t$. 
  \WHILE{$\Pi_{N+q}(x)$ is empty}
  \STATE Set $q = q+1$.
    \STATE Set $x^{(q)}=x$ and compute $u^{(q)}$ by solving \eqref{eq:gp_ucb}. 
    \WHILE{$t \leq t_{\text{infeasible}} + q\Delta t$}
\STATE Apply any admissible locally Lipschitz controller $\bar{\pi}$ with $\bar{\pi}(x^{(q)}) = u^{(q)}$ to the system. \label{codeline:admissibleinputtemporal} 
\ENDWHILE
\STATE Set $\mb{y}^{(q)} = \dot{x}^{(q)} + \mb{\xi}^{(q)}$.
\STATE Set $\mathcal{D}_{\Delta N} = \mathcal{D}_{\Delta N} \bigcap \{\mb{z}^{(q)}, \mb{y}^{(q)}\} $.
\STATE Update GP and $\Pi_{N+q}$.
  \ENDWHILE
\STATE Set $N = N+q$ and $\Delta N = \Delta N -q$.
  \ENDIF
 \end{algorithmic} 
 \end{algorithm}

\begin{theorem}
\label{thm:mainresult} 
Let \Cref{assumption:knownboundontimederivative,assumption:Lipschitzh,assp:robust_cbf_feasibility,assumption:noise,assumption:finiteRKHSnorm} hold, and let $x(0)$ be within the interior of the safe set $\mathcal{C}$, i.e., $x(0)\in \mathcal{C}\backslash \partial\mathcal{C}$. Furthermore, let $N$ denote the initial amount of data points used to train the GP, and choose
\[\eta \triangleq \left(\sum_{i=1}^n \frac{4L_h C_\beta C_i} {\epsilon \ln{\left(1 + \sigma_{i,\text{ns}}^{-2}\right)}} \right) ^{-\frac{2}{3(n+m+1)}} \frac{1}{3(n+m+1)}\]
where $C_i$ and $C_{\beta}$ are chosen as in \Cref{lemma:growthofgammai} and \Cref{lemma:growth_of_betaN}, respectively. Then, if we employ 
\Cref{alg:feasibility_recovering_algorithm_temporal} with
\begin{align*}
    \Delta N \geq & \eta^{-1} \left(\log\left(N\right)- \log\left(\eta\right)\right),
\end{align*}
with probability at least $1-\delta$, feasibility is recovered using a finite number of points before the state $x$ reaches $\partial\mathcal{C}$ and only samples the state space at most $\Delta N$ times, thus for any piecewise locally Lipschitz controller $\pi(x) \in \Pi_N(x)$ the system \eqref{eq:cl_dyn} safe with respect to $\mathcal{C}$ with probability at least $1 - \delta$. 
    
\end{theorem}
\begin{proof} 

Note that 
\Cref{alg:feasibility_recovering_algorithm_temporal} collects at least $\Delta N$ data points before leaving the safe set $\mathcal{C}$. Hence, it is sufficient to show that $\Pi_N(x)$ is non-empty for all $x\in \mathcal{X}$ after $\Delta N$ points have been collected. Safety with respect to $\mathcal{C}$ with probability at least $1-\delta$ then follows from \Cref{lemma:safetyifsetisnonempty}. 

We then show that $\Pi_N(x)$ is non-empty for all $x\in \mathcal{X}$ after at most $\Delta N$ points have been collected by contradiction. Let $\tilde{N}\triangleq N+q$. Note that 
\Cref{alg:feasibility_recovering_algorithm_temporal} only collects data points whenever $\Pi_{\tilde{N}}(x)$ is empty, implying
    \begin{align*}
& \frac{\partial h }{\partial x } (x^{(\tilde{N}+1)})\mu_{\tilde{N}}(x^{(\tilde{N}+1)}, u)  \\
<&L_h\beta_{\tilde{N}} \sqrt{\text{tr}\left(\mb{\Sigma}_{\tilde{N}}^2(x^{(\tilde{N}+1)},u)\right)}- \alpha(h(x^{(\tilde{N}+1)})) 
\end{align*}
holds for all $u\in\mathcal{U}$. Now, let
\begin{align*}
&u^{**}_{\tilde{N}+1} \triangleq  \\ &\arg\sup_{u \in \mathcal{U}} \frac{\partial h}{\partial x}\left(x^{(\tilde{N}+1)}\right)\left(f\left(x^{(\tilde{N}+1)}\right) + g\left(x^{(\tilde{N}+1)}\right)u\right),
\end{align*}
and note that, with probability at least $1-\delta$,
\begin{align}
\begin{split}
&\frac{\partial h }{\partial x } \left(x^{(\tilde{N}+1)}\right)\left(f\left(x^{(\tilde{N}+1)}\right) + g\left(x^{(\tilde{N}+1)}\right) u^{**}_{\tilde{N}+1} \right)\\
\leq&\frac{\partial h }{\partial x } \left(x^{(\tilde{N}+1)}\right)\mu_{\tilde{N}}\left(x^{(\tilde{N}+1)},u^{**}_{\tilde{N}+1} \right) \\&+ L_h\beta_{\tilde{N}} \sqrt{\text{tr}\left(\mb{\Sigma}_{\tilde{N}}^2\left(x^{(\tilde{N}+1)},u^{**}_{\tilde{N}+1} \right)\right)} \\
\leq&\frac{\partial h }{\partial x } \left(x^{(\tilde{N}+1)}\right)\mu_{\tilde{N}}\left(x^{(\tilde{N}+1)},u^{(\tilde{N}+1)}\right) \\&+ L_h\beta_{\tilde{N}} \sqrt{\text{tr}\left(\mb{\Sigma}_{\tilde{N}}^2\left(x^{(\tilde{N}+1)},u^{(\tilde{N}+1)}\right)\right)} .
   \end{split}
\end{align}

We then obtain
   \begin{align*}
        &- \alpha(h(x^{(\tilde{N}+1)})) \\
        >& \frac{\partial h }{\partial x } (x^{(\tilde{N}+1)}) \mu_{\tilde{N}}\left(x^{(\tilde{N}+1)},u^{(\tilde{N}+1)} \right) \\&- L_h\beta_{\tilde{N}} \sqrt{\text{tr}\left(\mb{\Sigma}_{\tilde{N}}^2(x^{(\tilde{N}+1)},u^{(\tilde{N}+1)})\right)} \\
        \geq &   \frac{\partial h }{\partial x } \left(x^{(\tilde{N}+1)}\right) \left(f\left(x^{(\tilde{N}+1)}\right) + g\left(x^{(\tilde{N}+1)}\right) u^{**}_{\tilde{N}+1} \right) \\
        &- 2L_h \beta_{\tilde{N}} \sqrt{\text{tr}\left(\mb{\Sigma}_{\tilde{N}}^2(x^{(\tilde{N}+1)},u^{(\tilde{N}+1)})\right)} \\
        \geq & -\alpha\left(h(x^{(\tilde{N}+1)})\right) +\epsilon  \\
        &- 2L_h\beta_{\tilde{N}}\sqrt{\text{tr}\left(\mb{\Sigma}_{\tilde{N}}^2(x^{(\tilde{N}+1)},u^{(\tilde{N}+1)})\right)}
    \end{align*}
i.e.,
\begin{align*}
    2L_h \beta_{\tilde{N}} \sqrt{\text{tr}\left(\mb{\Sigma}_{\tilde{N}}^2\left(x^{(\tilde{N}+1)},u^{(\tilde{N}+1})\right)\right)} \geq \epsilon.
\end{align*}
For simplicity of exposition, we introduce
\[\psi \triangleq \left( \sum_{i=1}^n \frac{4L_h} { \ln{\left(1 + \sigma_{i,\text{ns}}^{-2}\right)}} C_\beta C_i\right).\]
By summing up the posterior covariance terms over the collected data points and employing \Cref{lemma:growthofgammai,lemma:growth_of_betaN}, we then obtain
 \begin{align*}
        -\Delta N \epsilon  \geq &\sum_{q=1}^{\Delta N} - 2L_h \beta_{N+q}\sqrt{\text{tr}\left(\mb{\Sigma}_{N+q}^2(x^{({N}+q)},u^{(N+q)}_{{N}+q})\right)} \\
        \geq & - \sum_{q=1}^{\Delta N} \left(2L_h \beta_{N+q}\sum_{i=1}^n \frac{2} {\ln{\left(1 + \sigma_{i,\text{ns}}^{-2}\right)}} \gamma_{i,N+q} \right) \\
        \geq & - \psi \log(N+\Delta N)^{\frac{3}{2}(n+m+1)}.
        \end{align*}
    Hence,  \begin{align*}\exp\left(\left(\frac{\Delta N \epsilon}{\psi} \right)^{\frac{2}{3(n+m+1)}}\right) \leq N+\Delta N , 
        \end{align*}
holds for all $\Delta N>0$. Without loss of generality, we can substitute $\Delta N$ with $(\Delta N)^{\frac{3(n+m+1)}{2}}$, which yields
\begin{align*}&\exp\left(\Delta N \left(\frac{ \epsilon}{\psi} \right)^{\frac{2}{3(n+m+1)}}\right) \leq N+(\Delta N)^{\frac{3(n+m+1)}{2}} \\
\leq & (N+\Delta N)^{\frac{3(n+m+1)}{2}}, 
        \end{align*}
        i.e.,
        \begin{align*}&\exp\left(2 \Delta N \eta\right) 
\leq N+\Delta N.
        \end{align*}
    Since we have \begin{align*}
   \Delta N \geq & \eta^{-1}\left( \left(\log\left(N\right)- \log\left(\eta\right)\right)\right),
\end{align*}
this is a contradiction by  \Cref{lemma:exp_overtakes_linear}. By the same argument, 
\Cref{alg:feasibility_recovering_algorithm_temporal} collects at most $\Delta N$ points in the state space where $\Pi_{N+q}(x)$ is empty.

\end{proof}
\begin{remark}
    Although \Cref{alg:feasibility_recovering_algorithm_temporal} employs a temporal trigger to stipulate when exploration takes place, a trigger based on the distance to the boundary or the value of the CBF can also be employed, e.g., by dividing the CBF into $\Delta N$ segments and sampling whenever a new segment is reached.
\end{remark}
\section{Discussion} 

\subsection{Choice of $\Delta N$.}
Though the constants $\C_i$, $C_\beta$ and $\beta_i$ required by \Cref{thm:mainresult} can be computed explicitly \cite{Srinivas2012,chowdhury2017kernelized}, they can result in very conservative values for $\Delta N$. However, this is only a necessary criterion, as opposed to a sufficient one, and practice we can choose lower values of $\Delta N$. This is because  $\Delta N$ corresponds to the maximal amount of data that is required in order for $\Pi_{\tilde{N}}(x)$ to be non-empty for all $x \in \mathcal{C}$, and often only a subset of $\mathcal{C}$ is visited during control, i.e., it is sufficient for $\Pi_{\tilde{N}}(x)$ to be non-empty only for a subset of the safe set $\mathcal{C}$.

\subsection{Practicability of Online Exploration.}

A practical concern that may generally arise with control inputs $u$ that are geared towards exploring the state and input spaces is that the corresponding inputs may place a significant strain on the system and lead to undesirable behavior, e.g., if the input exhibits a high frequency and amplitude. However, computing the input by solving \eqref{eq:gp_ucb} corresponds to choosing the optimistically safest input under uncertainty. Hence it is reasonable to expect the corresponding input to be acceptable for the system, i.e., not damaging.


\begin{figure}\includegraphics[width = 0.99\columnwidth]{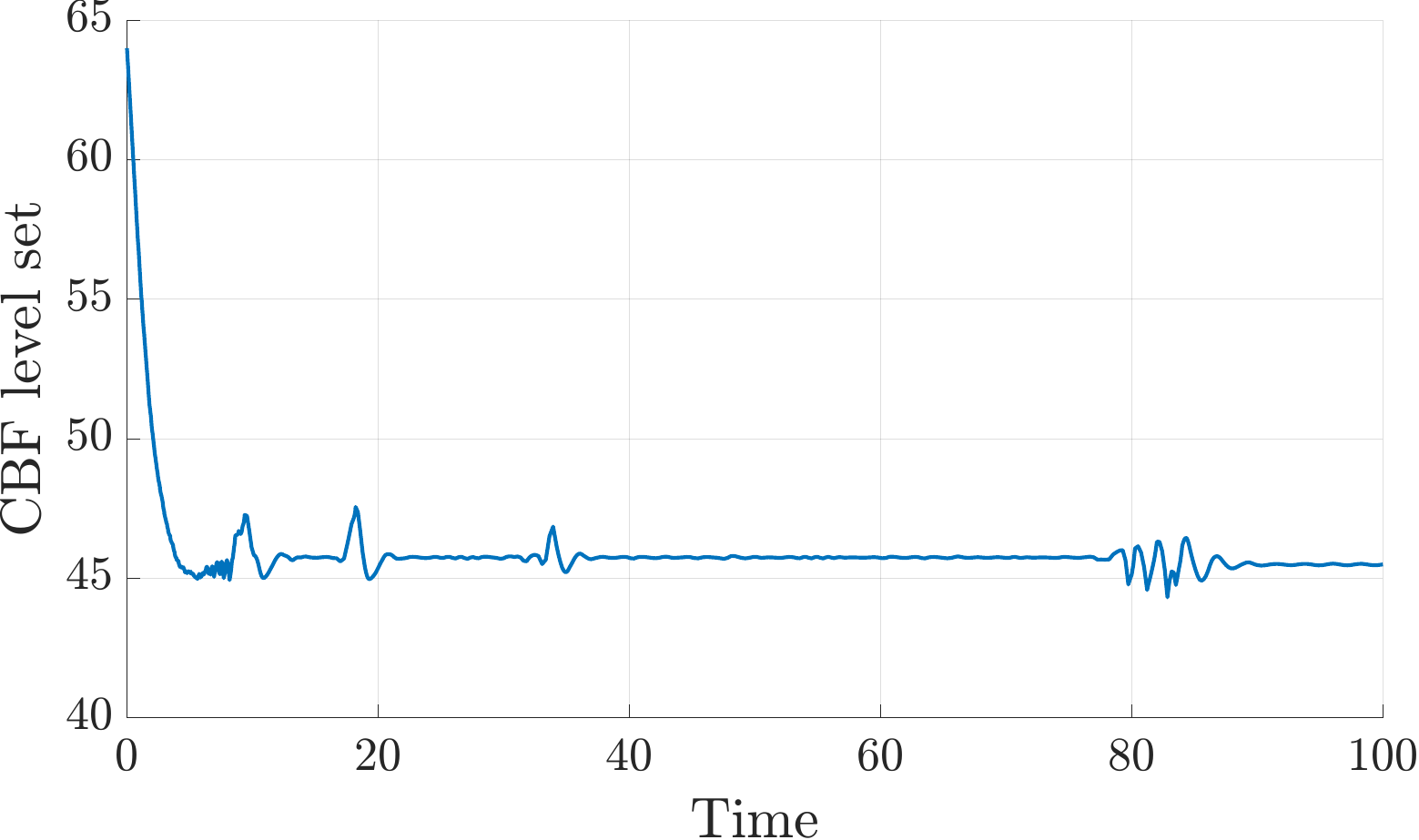}
\caption{Value of control barrier function $h(\mathbf{x})$ for cruise control example. Since no prior model is available for control, safety is obtained by efficiently learning a model online.} 
\label{fig:simu_run1_nomodel}
\end{figure}

\begin{figure}
\includegraphics[width = 0.99\columnwidth]{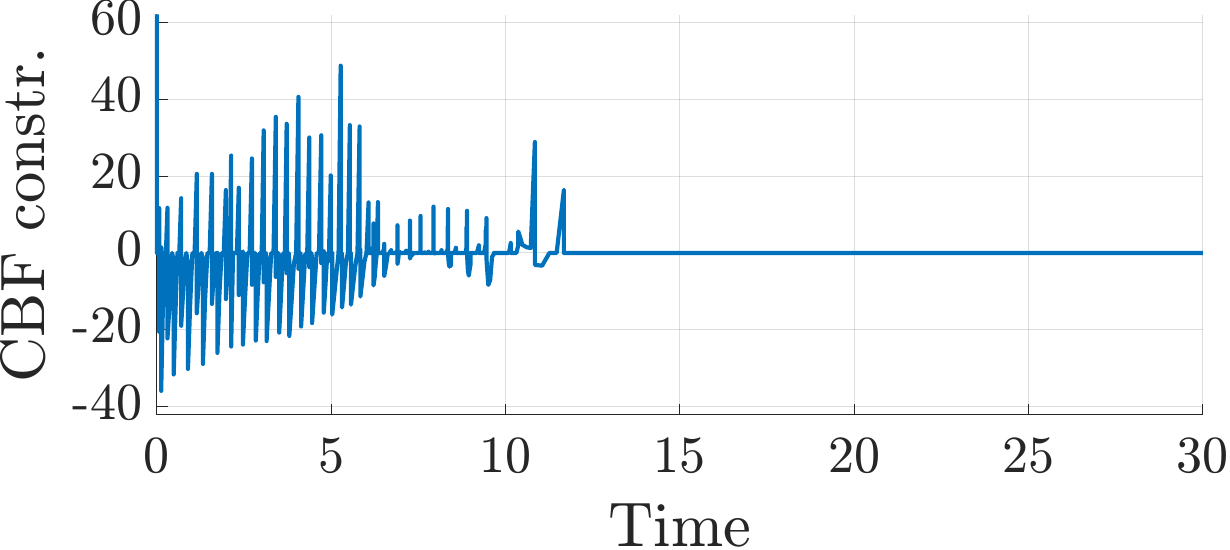}
\caption{Estimated (worst-case) time derivative of control barrier function for cruise control example. Spikes are due to training data set updates, leading to decreased model uncertainty. Positive values indicate infeasibility of the (GP-CBF-SOCP).} 
\label{fig:simu_run2_nomodel}
\end{figure}

\section{Numerical Validation}
\label{sect:numericalvalidatoin}

We now showcase how our approach performs using two numerical simulations. We start with a cruise control system, then present results for a quadrotor with ground dynamics, which is more complex. Note that, in the following, we assume to have either no prior model (cruise control example) or to know only the model component corresponding to the time-derivatives (quadrotor), which is insufficient to implement any state-of-the-art approach. We finish this section by comparing our approach to persistence of excitation. 

\subsection{Cruise Control}
\label{subsect:cruise_control}

Our approach is employed to learn the road vehicle model presented in \cite{castaneda2021pointwise} while simultaneously applying an adaptive cruise control system. In the following, we omit phyisical dimensions when describing the system model. The state space model is as in \eqref{eq:ol_dyn}, with unknown state-dependent functions
\begin{align}
\label{eq:state_space_model_cruise_ctrl}
    f(x) = 
    \begin{bmatrix}
        -\frac{1}{m}(\zeta_0 + \zeta_1 v + \zeta_2 v^2 ) \\
        v_0 - v
    \end{bmatrix}, \quad g(x) = \begin{bmatrix}
        0 \\
        \frac{1}{m}
    \end{bmatrix} 
\end{align}
and state $x= [v \ z]^\top $. Here $z$ denotes the distance between the ego vehicle and the target vehicle in front, $v$ denotes the ego vehicle speed, $m = 1650$ its mass, and $\zeta_0=0.2$, $\zeta_1=10$, $\zeta_2 =0.5$ are parameters that specify the rolling resistance. As a control barrier function, we employ $h(\mathbf{x}) = z - T_h v$, where $T_h=1.8$, which aims to maintain a safe distance between the ego vehicle and the vehicle in front. The nominal controller $\pi_\textrm{nom}(x)$ used for the (GP-CBF-SOCP) is a P-controller $\pi_\textrm{nom} = -10 (v-v_d)$, where $v_d=24$ 
corresponds to the desired velocity. We assume to have $N=10$ data points at the start of the simulation, which we employ exclusively to learn the lengthscales and signal variances of the covariance kernels by minimizing the posterior likelihood \cite{Rasmussen2006}. Though the amount of data $\Delta N$ stipulated by \Cref{thm:mainresult} yields strict theoretical guarantees, this represents only a sufficient, and not a necessary condition, and can be conservative in practice. Hence, to additionally showcase the practical applicability of our approach, we set $\Delta t = 10^{-5}$, which corresponds to a high rate for many practical applications.

We simulate the system for $100$ seconds. The CBF value can be seen in Fig. \ref{fig:simu_run1_nomodel}, the worst-case estimated value of the CBF constraint is depicted in Fig. \ref{fig:simu_run2_nomodel}. The CBF value $h(x)$ is always above zero, meaning that safety is always kept at all times. This is to be expected from \Cref{thm:mainresult}. The (CBF-SOCP) is infeasible during many simulation instances, particularly in the beginning, which leads to a high rate of exploratory inputs, obtained by solving \eqref{eq:gp_ucb}. However, feasibility is recovered after approximately $6$ seconds, after which a safe input can be obtained without further exploration. Note that the boundary of the safe set is not reached. This is because the proposed algorithm stops exploring when a safe controller can be computed, i.e., the (GP-CBF-SOCP) becomes feasible. Although closer proximity to the boundary of the safe set can be potentially achieved through exploration, we leave this to future work.

\begin{figure}\includegraphics[width = 0.99\columnwidth]{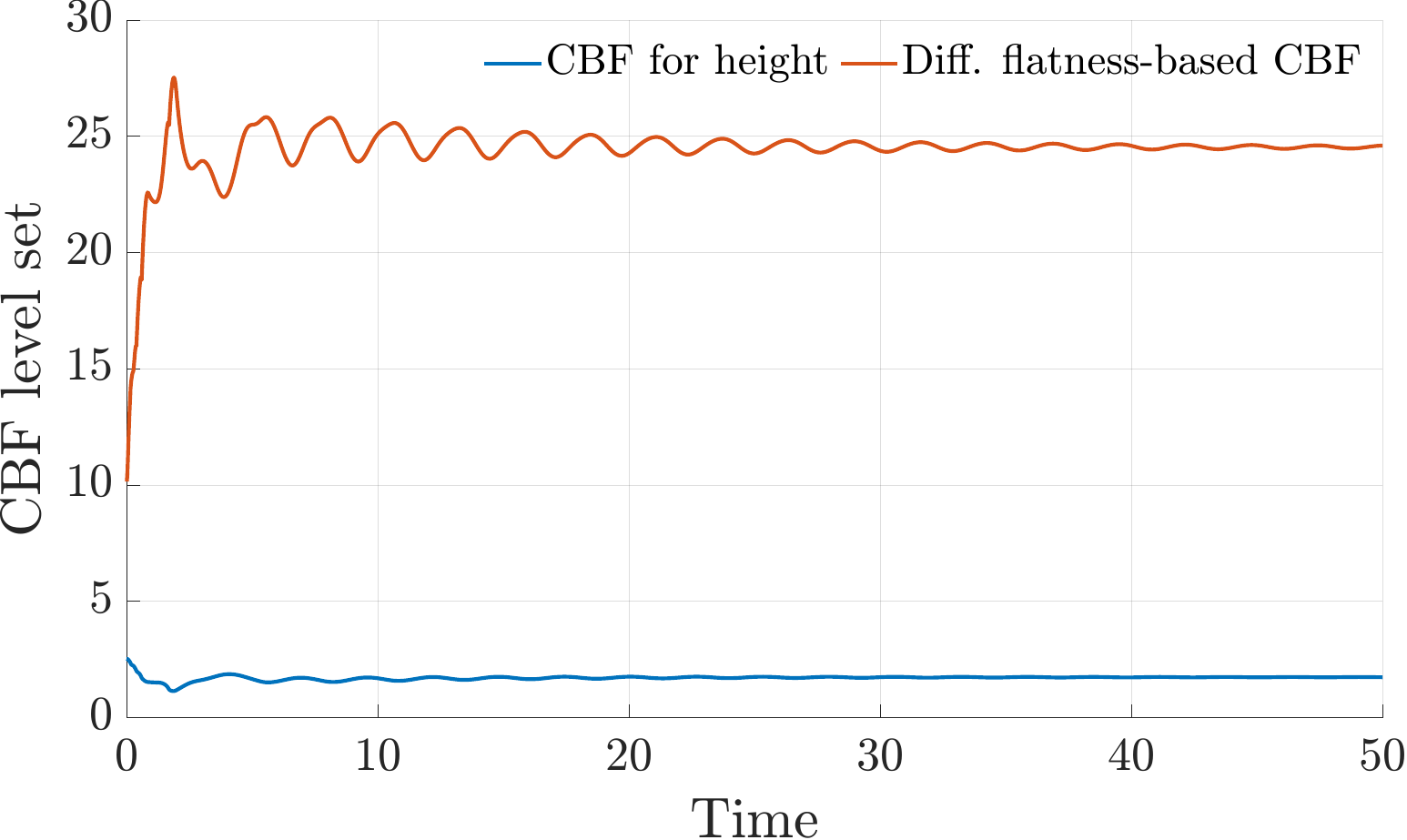}
\caption{Value of control barrier functions for quadrotor example..} 
\label{fig:sim_run_quadrotor_val}
\end{figure}

\begin{figure}
\includegraphics[width = 0.99\columnwidth]{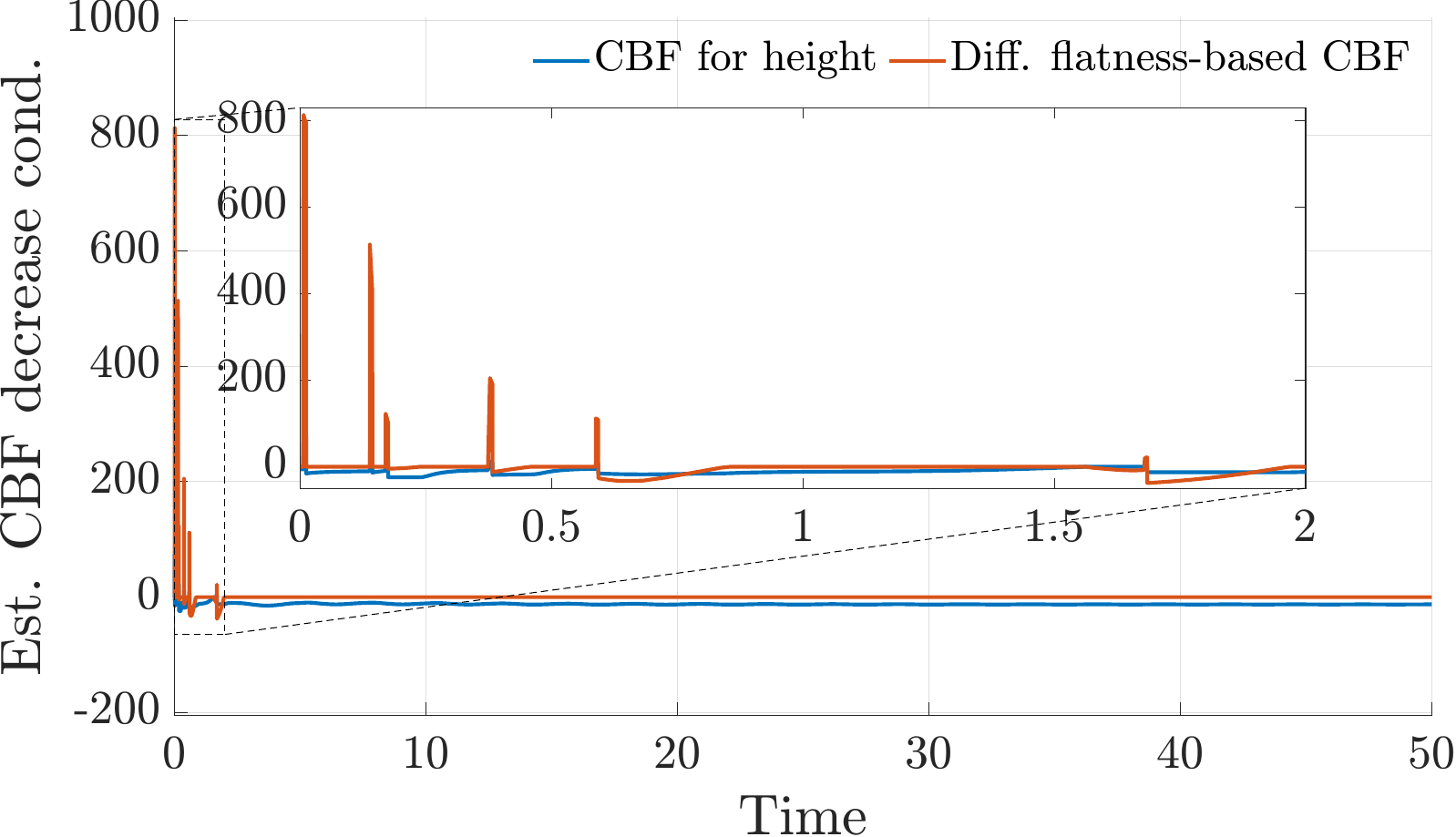}
\caption{Estimated (worst-case) time derivative of control barrier function for quadrotor example.} 
\label{fig:sim_run_quadrotor_worst_case}
\end{figure}

\begin{figure*}[t]
\centering
\begin{subfigure}[b]{0.49\textwidth}
\includegraphics[trim={0pt 0pt 0pt 0pt},clip,width = 0.9\columnwidth]{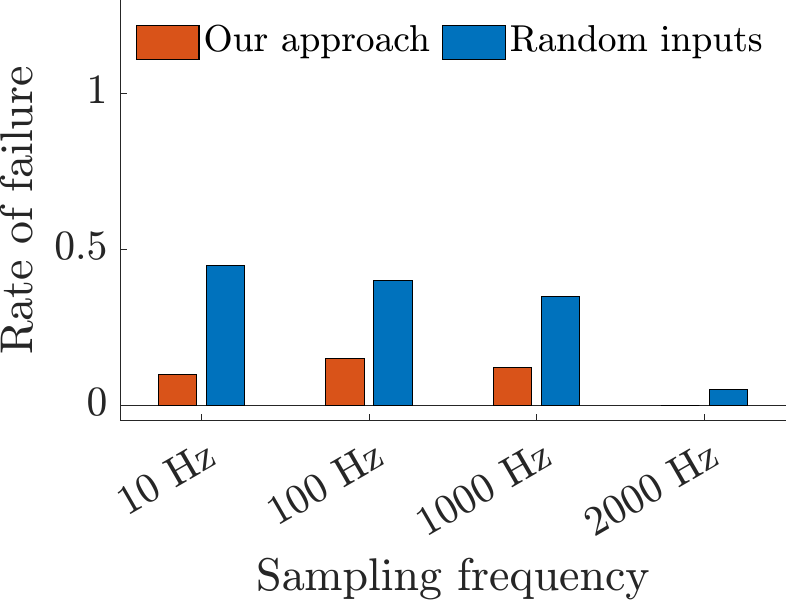}
\caption{Cruise control.} 
\label{fig:constr_satisf_uniform}
\end{subfigure}
\hfill
\begin{subfigure}[b]{0.49\textwidth}
\includegraphics[trim={0pt 0pt 0pt 0pt},clip,width = 0.9\columnwidth]{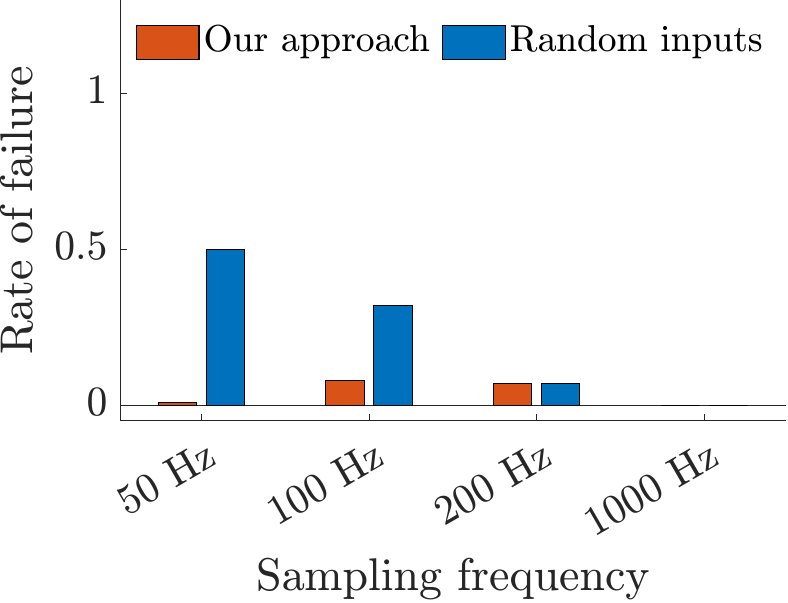}
\caption{Quadrotor.} 
\label{fig:cost_uniform}
\end{subfigure}
\caption{Rate of failure using our approach and random inputs for the cruise control (left) and quadrotor (right) settings.}
\label{fig:failures}
\end{figure*}

\subsection{Quadrotor}

We now showcase our approach using a numerical simulation of a quadrotor with ground dynamics. The quadrotor dynamics are specified by the functions
\begin{subequations}
\label{eq:state_space_model_quadrotor}
\begin{align}
    \dot{p} &= v, \quad & \dot{v} &= g_{\textup{gr}} e_z  + \zeta(p_z) R e_z T , \\ \dot{R} &= R [\omega]_\times 
    \end{align}
\end{subequations}
where $p\in\mathbb{R}^3$ is the global position, $v\in\mathbb{R}^3$ the global velocity, and $R\in \text{SO}(3)$ the system orientation. The parameter $g_\textup{gr} = 9.81$ denotes gravity, $[\; \cdot\; ]_\times: \mathbb{R}^3 \rightarrow \R^{3\times3}$ is the skew-symmetric mapping, and $e_z = [0 \ 0 \ 1]^\top$ is the unit-$z$ vector. The function $\zeta:\mathbb{R}_+ \rightarrow [0,1]$ models ground effects and takes the quadrotor height $p_z$ as an input variable. It is computed as \cite{danjun2015autonomous}
\begin{align}
    \zeta(p_z) = 1-\rho\left(\frac{r_{\text{rot}}}{4p_z}\right)^2,
\end{align}
where $\rho =  5$ and $r_{\text{rot}} = 0.09 $ is the rotor radious. The control inputs are $T$, the thrust acceleration, and $\omega \in  SO(3)$, the body-frame angle rates: 
\[
u= [T \ \omega ]^\top \in \R \times SO(3). 
\]

In this example, we employ two CBFs. The first is $h(\mathbf{x}) = 10(p_z - T_z v_z)$, where $T_z=0.1$, and is geared toward keeping altitude higher than zero. The second CBF utilizes the differential flatness of the quadrotor to synthesize a viable safe set. In particular we are concerned with restricting the system to safe positions defined as the 0-superlevel set of $h_p(p) = r^2 - \Vert p \Vert^2$. We extend $h_p$ to include velocities as $h_e(p,v) = \dot{h}_p(p,v) + \alpha h_p(p) $ for some $\alpha >0 $ to produce a relative degree 1 CBF for a double integrator system as in \cite{nguyen2016exponential}. To include orientation we add a rotation term to produce a CBF for the drone, $h(p,v,R) = h_e(p,v) - \lambda (1 - \frac{1}{2r}\frac{\partial h}{\partial p} R e_z) $, with some $\lambda \in (0, r^2/2)$ which shrinks the safe set to ensure that the thrust vector is pointing inwards whenever $h_e(p,v) = 0$. When computing the exploring input \eqref{eq:gp_ucb}, we alternate between CBFs. The nominal controller $\pi_\textrm{nom}(x)$ used for the (GP-CBF-SOCP) corresponds to a differentially flat controller, computed as in \cite{faessler2018differentially}, and we consider bounded thrust, with $\vert T\vert \leq 15000$. Note that $\omega \in SO(3)$ is already bounded. Similarly to the cruise control setting, we assume to have $N=10$ data points at the start of the simulation to learn the kernel hyperparameters and set $\Delta t = 10^{-5}$. 

We simulate the system for $50$ seconds. The CBF value can be seen in Fig. \ref{fig:sim_run_quadrotor_val}, the worst-case estimated value of the CBF constraint is depicted in Fig. \ref{fig:sim_run_quadrotor_worst_case}. Similarly to the cruise control case, instances, when the (GP-CBF-SOCP) is infeasible, are initially frequent, leading to a high data collection rate. After two seconds, feasibility is recovered, and safety is guaranteed by solving the (GP-CBF-SOCP).

\section{Comparison with Persistence of Excitation}
\Cref{thm:mainresult} states that it is \textit{sufficient} to collect data by applying \eqref{eq:gp_ucb} to the system in order to guarantee safety. However, other types of control inputs may also satisfy this requirement. In the following, we investigate how our approach performs compared to persistence of excitation. More specifically, we apply our method to the system with the following difference: instead of computing exploratory inputs by solving \eqref{eq:gp_ucb}, we sample the control inputs from a uniform distribution on $\mathcal{U}$, which corresponds to a persistently exciting signal \cite{anderson1985adaptive}. We again consider the adaptive cruise control and quadrotor settings and investigate how our approach and the random input-based one perform if the maximal sampling frequency $\frac{1}{\Delta t}$ is bounded. This is relevant, as many practical settings do not allow for arbitrarily high sampling frequencies.

We perform $100$ simulations with different initial conditions, uniformly sampled from a region within the safe set. We report how often each method fails, i.e., leads to a positive value for the CBF during the simulation. The average number of failures for both settings is shown in \Cref{fig:failures}. A non-zero failure rate is expected at low sampling frequencies since we cannot efficiently learn the system model if too little data is collected. However, as can be seen, our approach nonetheless performs better than the random control input-based approach. This is because the inputs applied to the system are geared towards recovering the feasibility of the (GP-CBF-SOCP), whereas random inputs are not. This also results in a higher data efficiency, as reflected in the average collected data, which is lower for our approach. At higher sampling rates, data efficiency is higher, leading to less collected data. This is because a higher sampling rate means that the elapsed time between data collection and model update is smaller, i.e., the model captures the true system more faithfully immediately after an update. This suggests that, if the maximal sampling frequency is low, then particular effort should be put into updating the GP as fast as possible.

\section{Conclusion}
\label{section:conclusion}

We have presented an online learning-based approach to recovering feasibility of a CBF-QP before reaching the boundary of a safe set, thus guaranteeing safety with high probability. This is achieved by leveraging tools commonly used in Bayesian optimization to devise an exploration approach that makes progress towards learning safe inputs. In future, we aim to apply the proposed approach to real-life and more complex systems. 

\bibliographystyle{IEEEtran}
\bibliography{references,cosner_main}

\end{document}